\documentclass{ceurart}

\usepackage[
  natbib = true,
    backend=biber,
    isbn=false,
    url=false,
    doi=false,
    eprint=false,
    style=numeric,
    sorting=nyt,
    sortcites = true
]{biblatex}
\AtEveryBibitem{%
	\ifentrytype{article}{
		\clearname{editor}
		\clearfield{issn}
        \clearfield{pages}
		\clearfield{urlyear}
	}{}
	\ifentrytype{inbook}{
		\clearname{editor}
		\clearfield{urlyear}
	}{}
	\ifentrytype{incollection}{
		\clearfield{isbn}
        \clearfield{pages}
		\clearfield{urlyear}
	}{}
	\ifentrytype{inproceedings}{
		\clearfield{address}
		\clearname{editor}
		\clearfield{eprint}
		\clearfield{isbn}
		\clearlist{location}
		\clearlist{publisher}
        \clearfield{volume}
        \clearfield{pages}
		\clearfield{urlyear}
		\clearfield{series}
	}{}
	\ifentrytype{proceedings}{
		\clearfield{isbn}
		\clearfield{urlyear}
	}{}
	\ifentrytype{misc}{
		\clearfield{isbn}
		\clearfield{urlyear}\clearfield{urlday}\clearfield{urlmonth} 
		\clearfield{month}
		\clearfield{eprint}
	}{}
}

\usepackage{todonotes}
\usepackage{amsmath}
\usepackage{amssymb}
\usepackage{amsthm}
\usepackage[linesnumbered,ruled,vlined]{algorithm2e}
\DontPrintSemicolon
\SetKwInOut{Input}{Input}
\SetKwInOut{Output}{Output}
\SetKwInOut{Data}{Data}
\SetAlFnt{\footnotesize}
\SetAlCapFnt{\small}
\SetAlCapNameFnt{\small}
\SetArgSty{textnormal}

\usepackage{cleveref}

\newtheorem{definition}{Definition}[section]
\newtheorem{theorem}{Theorem}[section]
\newtheorem{lemma}{Lemma}[section]
\newtheorem{example}{Example}[section]

\Crefname{lemma}{Lemma}{Lemmas}
\crefname{lemma}{Lemma}{Lemmas}
\usepackage{subcaption}
\usepackage{booktabs}
\usepackage{float}
\usepackage{array}
\usepackage{rwthcolors}

\DeclareMathOperator{\ordOp}{ord}
\newcommand\ord[2]{\ordOp_{#1}(#2)}

\newcommand\vdeg[2]{\deg_{#1}(#2)}

\DeclareMathOperator{\coeffOp}{coeff}
\newcommand\coeff[2]{\coeffOp_{#1}(#2)}
\DeclareMathOperator{\ldcfOp}{ldcf}
\newcommand\ldcf[2]{\ldcfOp_{#1}(#2)}
\DeclareMathOperator{\discOp}{disc}
\newcommand\disc[2]{\discOp_{#1}(#2)}
\DeclareMathOperator{\resOp}{res}
\newcommand\res[3]{\resOp_{#1}(#2,#3)}

\DeclareMathOperator{\realRootsOp}{rroots}
\newcommand\realRoots[1]{\realRootsOp(#1)}

\DeclareMathOperator{\irootOp}{root}
\newcommand\iroot[3]{\irootOp_{#1}^{#2,#3}}

\newcommand\rationals[0]{\mathbb{Q}}
\newcommand\reals[0]{\mathbb{R}}

\newcommand\naturals[0]{\mathbb{N}}

\newcommand\posints[0]{\mathbb{N}_{>0}}
\newcommand\projs[0]{\mathbb{P}}

\newcommand\proj[2]{\ensuremath{#1_{#2}}}

\newcommand\Isymb[0]{\textnormal{\texttt{I}}}

\newcommand\range[2]{\ensuremath{[#1..#2]}}
\newcommand\prange[1]{\ensuremath{[#1]}}

\newif\ifextended\extendedfalse

\begin{document}

\copyrightyear{2025}
\copyrightclause{Copyright for this paper by its authors.
  Use permitted under Creative Commons License Attribution 4.0
  International (CC BY 4.0).}

\conference{SC² 2025: Satisfiability Checking and Symbolic Computation}

\title{Projective Delineability for Single Cell Construction}

\tnotemark[1]
\tnotetext[1]{E.~Ábrahám and J.~Nalbach are supported by the
German Research Foundation (DFG)
as part of AB 461/9-1 (SMT-ART) %
and J.~Nalbach by RTG 2236 (UnRAVeL).
Computing resources were granted by RWTH Aachen University under project rwth1560.
P.~Mathonet, L.~Michel and N.~Zenaïdi are supported by the FNRS-DFG PDR Weaves (SMT-ART) grant 40019202.
M.~England and J.~Davenport are supported by the UKRI EPSRC DEWCAD Project (grant EP/T015748/1 and EP/T015713/1 respectively).
This collaboration was enabled by the project EuroProofNet, supported by European Cooperation in Science and Technology as COST Action CA20111.
The views expressed in this paper are those of the authors and do not reflect the official policy or position of the U.S. Naval Academy, Department of the Navy, the Department of Defense, or the U.S. Government.}

\author[1]{Jasper Nalbach}[orcid=0000-0002-2641-1380,email=nalbach@cs.rwth-aachen.de]
\cormark[1]
\author[2]{Lucas Michel}[orcid=0000-0002-4115-7296]
\author[1]{Erika {\'A}brah{\'a}m}[orcid=0000-0002-5647-6134]
\author[1]{Christopher W. Brown}[orcid=0000-0001-8334-0980]
\author[1]{James H. Davenport}[orcid=0000-0002-3982-7545]
\author[1]{Matthew England}[orcid=0000-0001-5729-3420]
\author[1]{Pierre Mathonet}[orcid=0000-0001-9199-8643]
\author[1]{Naïm Zénaïdi}[orcid=0009-0005-5520-5147]

\address[1]{RWTH Aachen University, Aachen, Germany}
\address[2]{University of Liège, Liège, Belgium}
\address[3]{United States Naval Academy, Annapolis, Maryland, USA}
\address[4]{University of Bath, Bath UK}
\address[5]{Coventry University, Coventry, UK}

\cortext[1]{Corresponding author.}

\begin{abstract}
The cylindrical algebraic decomposition (CAD) is the only complete method used in practice for solving problems like quantifier elimination or SMT solving related to real algebra, despite its doubly exponential complexity. Recent exploration-guided algorithms like \emph{NLSAT}, \emph{NuCAD}, and \emph{CAlC} rely on CAD technology but reduce the computational effort heuristically. \emph{Single cell construction} is a paradigm that is used in each of these algorithms.

The central property on which the CAD algorithm is based is called delineability. Recently, we introduced a weaker notion called \emph{projective delineability} which can require fewer computations to guarantee, but needs to be applied carefully. This paper adapts the single cell construction for exploiting projective delineability and reports on experimental results.
\end{abstract}


\maketitle

\section{Introduction}

The \emph{cylindrical algebraic decomposition (CAD)} method enables reasoning about formulas in \emph{real algebra} and is implemented in various tools for \emph{quantifier elimination} like \texttt{QEPCAD} \cite{brown2003}, \texttt{Redlog} \cite{seidl2003}, \texttt{Mathematica} \cite{strzebonski2014}, and \texttt{Maple} \cite{chen2009}, and \emph{satisfiability-modulo-theories (SMT) solving}, like \texttt{z3} \cite{demoura2008}, \texttt{cvc5} \cite{barbosa2022}, \texttt{yices2} \cite{dutertre2014}, and \texttt{SMT-RAT} \cite{corzilius2015}. Despite its doubly exponential complexity, it is the most widely used complete method for these problems.

The CAD method decomposes the real space into a finite number of connected sets (called \emph{cells}) such that the input set of polynomials have invariant sign in each cell. Although such a decomposition allows for reasoning about the formula, it is usually finer than needed for the task at hand. Algorithms like \emph{NLSAT} \cite{jovanovic2012}, \emph{NuCAD} \cite{brown2015nucad}, and \emph{CAlC} \cite{abraham2021} reduce the computational effort by computing only a set of cells where the input formula is \emph{truth-invariant} that together cover the real space rather than decompose it. These savings are achieved by using the Boolean structure and relation symbols to determine which polynomials are relevant in a certain part (determined by some sample point) of the space, and using the shape of the varieties to reduce the computation steps.  This leads to cells that are faster to compute and which are usually larger than the cells from the CAD.

More specifically, the CAD algorithm iteratively computes \emph{projection polynomials} to eliminate variable by variable. These consist of \emph{resultants}, \emph{discriminants}, and \emph{(leading) coefficients} which are selected to maintain the \emph{delineability} of the polynomials, a property which allows for work at a sample point to be generalized to a wider cell. The \emph{single cell construction} \cite{brown2013,brown2015,nalbach2024levelwise} paradigm is the foundation of the above algorithms and is able to reduce the amount of resultants and discriminants needed to maintain delineability. 

This paper investigates when we can go further and leave out leading coefficients from the projection. Towards this goal, in prior theoretical work, we proposed to relax delineability to \emph{projective delineability} \cite{michel2024}: a property that does not require leading coefficients to be maintained.  We now report on the embedding of projective delineability within single cell construction, and its impact in experimental results.  
We continue by recalling the single cell construction in Section \ref{sec:preliminaries} and projective delineability in \Cref{sec:ProjDel}.  Then in \Cref{sec:main} we present the modification of the single cell construction to use projective delineability, in \Cref{sec:experiments} we report an experimental evaluation from our implementation, and finally we conclude in \Cref{sec:conclusion}.

\section{Preliminaries}
\label{sec:preliminaries}

We introduce key background following \cite{nalbach2024mergingcells} (full details in preliminaries of \cite{nalbach2024levelwise}).

Let $\naturals$, $\posints$, $\rationals$, and $\reals$ denote the sets of all natural (incl. $0$), positive integer, rational, and real numbers respectively. 
For $i,j \in \naturals$ with $i<j$, we define
$\range{i}{j} = \{ i,\ldots,j \}$ and $\prange{i} = \range{1}{i}$.
For $i,j \in \posints, j \leq i$ and $r \in \reals^i$, we denote by $r_j$ the $j$-th component of $r$ and by $r_{\prange{j}}$ the vector $(r_1,\ldots,r_j)$.
Let $f,g: D \to E$ and let $<$ be a total order on $E$ and $\sim \in \{ <,=,\neq \}$. We write \emph{$f\sim g$ on $D$} if $f(d)\sim g(d)$ for all $d \in D$. Note that $f\neq g$ on $D$ is not ``not $f= g$ on $D$''.

We work with the \emph{variables} $x_1,\ldots,x_n$ with $n\in\posints$ under a fixed \emph{ordering} $x_1 \prec x_2 \prec ... \prec x_n$.  
A \emph{polynomial} is built from a set of variables and numbers from $\rationals$ using addition and multiplication.   
We use $\rationals[x_1,\ldots,x_i]$ to denote \emph{multivariate} polynomials in those variables.
A polynomial $p$ \emph{is of level $j$} if $x_j$ is the largest variable in $p$ with non-zero coefficient.

Let $i \in \prange{n}$ and $p,q\in \rationals[x_1,\ldots,x_i]$ of level $i$.
For $j\in \prange{i}$ and $r=(r_1,\ldots,r_j)\in\reals^j$ we write $p(r,x_{j+1},\ldots,x_{i})$ for the polynomial $p$ after substituting $r_1,\ldots,r_j$ for $x_1,\ldots,x_j$ in $p$ and indicating the remaining free variables in $p$.
We use $\realRoots{p} \subseteq \reals^i$ to denote the set of \emph{real roots of $p$}, $\vdeg{x_j}{p}$ to denote the \emph{degree of $p$ in $x_j$}, $\coeff{x_j}{p}$ for the set of \emph{coefficients of $p$ in $x_j$}, $\ldcf{x_j}{p}$ for the \emph{leading coefficient of $p$ in $x_j$}, %
 $\disc{x_j}{p}$ to denote the \emph{discriminant of $p$ with respect to $x_j$}, and $\res{x_j}{p}{q}$ to denote the \emph{resultant of $p$ and $q$ with respect to $x_j$}.
Let $r \in \reals^{i-1}$, then $p$ \emph{is nullified on $r$} if $p(r,x_{i}) = 0$.

A \emph{constraint} $p \sim 0$ compares a polynomial $p \in \rationals[x_1, \ldots, x_i]$ to zero using a relation symbol ${\sim \in \{=,\neq,<,>,\leq,\geq\}}$, and has \emph{solution set} $\{r\in\reals^i \mid p(r)\sim 0\}$.
A subset of $\reals^i$ for some $i \in \prange{n}$ is called \emph{semi-algebraic} if it is the solution set of a Boolean combination of polynomial constraints. A \emph{cell} is a non-empty connected subset of $\reals^i$ for some $i \in \prange{n}$. 
A cell $R$ is called \emph{simply connected} if any loop in $R$ can be continuously contracted to a point.
A polynomial $p \in \rationals[x_1, \ldots, x_i]$ is \emph{sign-invariant on a set $R \subseteq \reals^i$} if the sign of $p(r)$ is the same for all $r \in R$.

Given $i,j \in \posints$ with $j<i$, we define the \emph{projection of a set $R \subseteq \reals^i$ onto $\reals^{j}$} by $\proj{R}{\prange{j}}=\{(r_1,\ldots,r_j) \mid \exists r_{j+1},\ldots,r_i.\; (r_1,\ldots,r_i)\in R\}$.
Given a cell $R \subseteq \reals^i$, $i \in \prange{n}$ and continuous functions $f,g: R \to \reals$, we define the sets $R \times f = {\{ (r,f(r)) \mid r\in R \}}$ and ${R \times (f,g)} = {\{ (r,r_{i+1}) \mid r\in R , r_{i+1}\in (f(r),g(r))\}}$ (${R \times (-\infty,g)}$, ${R \times (f,+\infty)}$ analogously).

If $U \subseteq \reals^i$ is open, a function $f: U \to \reals^n$ is \emph{analytic} if each component of $f$ has a multiple power series representation around each point of $U$. An $i$-dimensional \emph{analytic submanifold} of $\reals^n$ is a non-empty subset $R \subseteq \reals^n$ locally parametrized by coordinates through analytic functions $f: U\subseteq\reals^i\to \reals^n$. A function $f$ between analytic manifolds $R$ and $R'$ is analytic if locally it has an expression in (analytic) coordinates which is analytic (see also \cite{krantz2002}).
Let $p \in \rationals[x_1, \ldots, x_n]$ be a polynomial and $r \in \reals^n$ be a point. Then the \emph{order of $p$ at $r$}, $\ord{r}{p}$, is defined as the minimum $k$ such that some partial derivative of total order $k$ of $p$ does not vanish at $r$ (and $+\infty$ if $p=0$).
We call $p$ \emph{order-invariant on $R \subseteq \reals^n$} if $\ord{r}{p} = \ord{r'}{p}$ for all $r,r' \in R$ (for details see \cite{mccallum1985}).

\subsection{CAD and Single Cell Construction}
\label{sec:preliminaries:scc}

A \emph{cylindrical algebraic decomposition} (CAD) \cite{collins1975,mccallum1985,mccallum1998} is a decomposition $\mathcal{C}$ of $\reals^n$ such that each cell $R \in \mathcal{C}$ is semi-algebraic and \emph{locally cylindrical} (i.e. can be described as the solution set of $\psi_1(x_1) \wedge \psi_2(x_1,x_2) \wedge \psi_n(x_1,\ldots,x_n)$ where $\psi_i$ is one of $x_i = \theta(x_1,\ldots,x_{i-1})$ or $\theta_l(x_1,\ldots,x_{i-1}) < x_i < \theta_u(x_1,\ldots,x_{i-1})$ or $\theta_l(x_1,\ldots,x_{i-1}) < x_i$ or $x_i < \theta_u(x_1,\ldots,x_{i-1})$ for some continuous functions $\theta,\theta_l,\theta_u$), and $\mathcal{C}$ is \emph{cylindrically arranged} (i.e. either $n=1$ or $\{ \proj{R}{\prange{n-1}} \mid R \in \mathcal{C} \}$ is a cylindrically arranged decomposition of $\reals^{n-1}$).
The shape of such a CAD allows reasoning about properties of (sets of) polynomials computationally. In particular, it is called \emph{sign-invariant for a set of polynomials $P \subseteq \rationals[x_1,\ldots,x_n]$} if each $p \in P$ is sign-invariant on each $R \in \mathcal{C}$. A sign-invariant CAD for $P$ is computed recursively: to describe the cells' boundaries for $x_n$, we first compute the underlying decomposition by a projection operation resulting in a set $P' \subseteq \rationals[x_1,\ldots,x_{n-1}]$ whose sign-invariant CAD will describe the first $n-1$ levels of the cells of the sign-invariant CAD of $P$.

The \emph{single cell construction} \cite{brown2015,nalbach2024levelwise} computes, given a set of polynomials $P \subseteq \rationals[x_1,\ldots,x_n]$ and a sample point $s \in \reals^n$, a locally cylindrical cell $R \subseteq \reals^n$ such that $s \in R$ and such that $P$ is sign-invariant on $R$.
In the rest of this section, we introduce the \emph{levelwise method} \cite{nalbach2024levelwise} for single cell construction.

\paragraph{Delineability.}  Delineability of a polynomial on some cell means that its variety can be described by continuous functions which are nicely ordered over that cell. This allows us to reason about the polynomial's roots using these functions.

\begin{definition}[Delineability \cite{mccallum1998}]
    \label{def:delineability}
    Let $i \in \naturals$, $R \subseteq \reals^{i}$ be a cell, and $p \in \rationals[x_1, \ldots, x_{i+1}] \setminus \{ 0 \}$. %
    Polynomial $p$ is \emph{delineable on $R$} if there exist continuous functions $\theta_1, \ldots, \theta_k: R \to \reals$, $k\in\naturals$, such that:
    \begin{itemize}
        \item $\theta_1 < \cdots < \theta_k$;
        \item the set of real roots of $p(r,x_{i+1})$ is $\{ \theta_1(r), \ldots, \theta_k(r) \}$ for all $r \in R$; and %
        \item there exist constants $m_1,\ldots,m_k \in \posints$ such that for all $r \in R$ and all $j\in\prange{k}$, the multiplicity of the root $\theta_j(r)$ of $p(r,x_{i+1})$ is $m_j$.
    \end{itemize}

    The $\theta_j$ are called \emph{real root functions of $p$ on $R$}. The sets $R \times \theta_j$ are called \emph{sections of $p$ over $R$}. %

    \emph{Analytic delineability} is similar, but $R$ is a connected analytic submanifold of $\reals^i$ and the real root functions are analytic.
\end{definition}

\noindent
The following gives a projection to obtain a cell where a polynomial is delineable.

\begin{theorem}[Delineability of a Polynomial {\cite[Thm. 2]{mccallum1998}}, {\cite[Thm. 3.1]{brown2001}}]
    \label{thm:delineable}
    Let $i \in \naturals$, $R \subseteq \reals^{i}$ be a connected analytic submanifold, $p \in \rationals[x_1, \ldots, x_{i+1}]$ of level $i+1$. Assume that $p$ is not nullified at any point in $R$, $\disc{x_{i+1}}{p}$ is not the zero polynomial and is order-invariant on $R$, and $\ldcf{x_{i+1}}{p}$ is sign-invariant on $R$.  
    Then $p$ is analytically delineable on $R$ and is order-invariant on its sections over $R$.
\end{theorem}

\noindent
Note that the discriminant of an irreducible polynomial is not the zero polynomial; in our algorithm, we replace each polynomial by its irreducible factors. 

\paragraph{Root Orderings.}

Once we can describe the roots of individual polynomials by ordered root functions on the underlying cell, we can reason about intersections of graphs of root functions from different polynomials, e.g. ensure that two root functions remain in the same order on the underlying cell.

\begin{theorem}[Lifting of Pairs of Polynomials {\cite[Thm. A.1]{nalbach2024levelwise}}]
    \label{thm:irordering}
    Let $i \in \naturals$, $R \subseteq \reals^{i}$ be a connected analytic submanifold, $s \in R$, and $p_1,p_2 \in \rationals[x_1, \ldots, x_{i+1}]$ of level $i+1$. Assume $p_1$ and $p_2$ are analytically delineable on $R$ and $\res{x_{i+1}}{p_1}{p_2}$ is not the zero polynomial and is order-invariant on $R$.  
    Let $\theta_1,\theta_2: R \to \reals$ be real root functions of $p_1$ and $p_2$ on $R$ respectively, and $\sim \in \{ <,= \}$ such that $\theta_1(s) \sim \theta_2(s)$.
    Then $\theta_1 \sim \theta_2$ on $R$.
\end{theorem}

Note that the resultant of two coprime (and irreducible) polynomials is not the zero polynomial.

To maintain that two real root functions $\theta_1$ and $\theta_2$ stay in the same order on $R$, we could also exploit transitivity using another root function $\theta_3$, e.g. $\theta_1 < \theta_3$ on $R$ and $\theta_3 < \theta_2$ on $R$ implies $\theta_1 < \theta_2$ on $R$. The work in \cite{nalbach2024levelwise} generalizes this idea to orderings on a set of root functions. This allows for flexibility in the choice of resultants which we compute to maintain certain invariance properties, potentially avoiding the computation of expensive resultants. 

\paragraph{Single Cell Construction.}  
Given a set of polynomials $P \subseteq \rationals[x_1,\ldots,x_i]$ and a sample $s \in \reals^i$, we compute and sort the real roots of $p(s_{\prange{i{-}1}},x_i),\ p \in P$. We determine the greatest root $\xi_\ell \in \reals$ below (or equal to) $s_i$ and the smallest root $\xi_u \in \reals$ above (or equal to) $s_i$. If they do not exist, we use $-\infty$ and $+\infty$ respectively. We now aim to describe the bounds of the cell $R' \subseteq \reals^i$ to be constructed by root functions of some polynomials in $P$; for that, we assume that all polynomials in $P$ are delineable on the underlying cell $R = \proj{R'}{\prange{i-1}}$. Let $\theta_\ell$ and $\theta_u$ be real root functions of polynomials in $P$ such that $\theta_\ell(s_{\prange{i{-}1}})=\xi_\ell$ and $\theta_u(s_{\prange{i{-}1}})=\xi_u$. The bounds on $x_i$ are described by the \emph{symbolic interval} $(\theta_\ell,\theta_u)$ (whose bounds depend on $x_1,\ldots,x_{i-1}$) if $\theta_\ell(s_{\prange{i{-}1}}) < \theta_u(s_{\prange{i{-}1}})$ or $\theta_\ell$ if $s_i = \xi_\ell = \xi_u$. Now, we use root orderings to make sure that $\theta_\ell < \theta_u$ on $R$ (if applicable) and each $p \in P$ is sign-invariant in $R \times (\theta_\ell,\theta_u)$ resp. $R \times \theta_\ell$:

\begin{theorem}[Root Ordering for Sign Invariance {\cite{nalbach2024levelwise}}]
    \label{thm:invariant}
    Let $i \in \naturals$, $R \subseteq \reals^{i}$ be connected, and $p,p_\ell,p_u \in \rationals[x_1, \ldots, x_{i+1}]$ of level $i+1$.
    Assume that $p$, $p_\ell$, $p_u$ are delineable on $R$. Let $\theta_\ell,\theta_u: R \to \reals$ be real root functions of $p_\ell$ and $p_u$ on $R$ respectively.
    
    \begin{itemize}
        \item If $\theta_\ell < \theta_u$ on $R$, and for each real root function $\theta$ of $p$ on $R$ it holds $\theta \sim \theta_\ell$ on $R$ for some $\sim \in \{ <, =\}$ or $\theta_u \sim \theta$ on $R$ for some $\sim \in \{ <, =\}$; then $p$ is sign-invariant on $R \times (\theta_\ell,\theta_u)$.
        \item If for each real root function $\theta$ of $p$ on $R$ it holds $\theta_u \sim \theta$ on $R$ for some $\sim \in \{ <, =\}$; then $p$ is sign-invariant on $R \times (-\infty,\theta_u)$.
        \item If for each real root function $\theta$ of $p$ on $R$ it holds $\theta \sim \theta_\ell$ on $R$ for some $\sim \in \{ <, =\}$; then $p$ is sign-invariant on $R \times (\theta_\ell,+\infty)$. 
        \item If for each real root function $\theta$ of $p$ on $R$ it holds either $\theta < \theta_\ell$ on $R$, or $\theta_\ell < \theta$ on $R$, or $\theta = \theta_\ell$ on $R$; then $p$ is sign-invariant on $R \times \theta_\ell$.
        \item If there is no real root function $\theta$ of $p$ on $R$; then $p$ is sign-invariant on $R \times \reals$.
    \end{itemize}
\end{theorem}

\begin{algorithm}[tb]
    \caption{\texttt{single\_cell\_construction($P$,$s$)}}
    \label{alg:scc}
    \SetKwInOut{Input}{Input}
    \SetKwInOut{Output}{Output}
    \DontPrintSemicolon
    \Input{finite $P \subseteq \rationals[x_1,\ldots,x_n]$, $s \in \reals^n$}
    \Output{Symbolic intervals $\Isymb_1,\ldots,\Isymb_n$ for $x_1,\ldots,x_n$ describing a sign-invariant cell for $P$ containing $s$}

    \ForEach{$i=n,\ldots,1$}{
        $P_i := \{ p \in P \mid p \text{ is of level } i \}$, $P := P \setminus P_i$ \;
        determine the set of indexed roots $\Theta = \{\theta_1,\ldots,\theta_k\}$ of all $p \in P_i$ that are defined at $s_{\prange{i{-}1}}$ such that $\theta_1(s_{\prange{i{-}1}}) \leq \ldots \leq \theta_k(s_{\prange{i{-}1}})$ \; \label{alg:line:roots}

        \tcp{Determine symbolic interval $\Isymb_i$}
        \lIf{$s_i = \theta_j(s_{\prange{i{-}1}})$ for some $j$}{
            $\Isymb_i := \theta_j$ 
        }
        \lElseIf{$\theta_j(s_{\prange{i{-}1}}) < s_i < \theta_{j+1}(s_{\prange{i{-}1}})$ for some $j$}{
            $\Isymb_i := (\theta_j,\theta_{j+1})$ 
        }
        \lElseIf{$s_i < \theta_{1}(s_{\prange{i{-}1}})$ }{
            $\Isymb_i := (-\infty,\theta_1)$ 
        }
        \lElseIf{$\theta_{k}(s_{\prange{i{-}1}}) < s_i$}{
            $\Isymb_i := (\theta_k,+\infty)$ 
        }
        \lElse{
            $\Isymb_i := (-\infty,+\infty)$ 
        }

        \ForEach{$p \in P_i$}{\label{alg:line:mod-begin}\label{alg:line:loop}
            \tcp{Ensure order invariance for each polynomial}
            \lIf{$p(s_{\prange{i{-}1}},x_i) = 0$}{\Return{FAIL}} \label{alg:line:incomplete}
            $P := P \cup \{ c \}$ for some $c \in \coeff{x_i}{p}$ such that $c(s) \neq 0$ \; \label{alg:line:nonnull}
            $P := P \cup \{ \disc{x_i}{p},\ldcf{x_i}{p} \}$ \tcp*{delineability, \Cref{thm:delineable}} \label{alg:line:del}
        }
        \textbf{choose} ${\preceq} \subseteq \Theta^2$ s.t. its reflexive and transitive closure $\preceq^{rt}$ is a partial order on $\Theta$ with $\theta_\ell \preceq^{rt} \theta_u$ (if $\Isymb_i = (\theta_\ell,\theta_u)$) and ensures sign-invariance of each $p \in P_i$ by \Cref{thm:invariant} \; \label{alg:line:irordering}
        $P := P \cup \{ \res{x_i}{p}{p'} \mid (\iroot{x_i}{p}{j},\iroot{x_i}{p'}{j'}) \in {\preceq} \}$ \tcp*{\Cref{thm:irordering}}
        \label{alg:line:mod-end}
    }

    \Return{$\Isymb_1,\ldots,\Isymb_n$}
\end{algorithm}

The single cell construction is given in \Cref{alg:scc}.  
In \Cref{alg:line:roots}, we determine witnesses for the real root functions of a polynomial $p$ of level $i$ on $R$ (the underlying cell to be constructed). Given some $j \in \posints$, an \emph{indexed root} is a partial function $\iroot{x_{i}}{p}{j}: \reals^{i} \to \reals$ that maps $s \in \reals^{i-1}$ to the $j$-th real root of $p(s,x_{i})$ if it exists. Given a cell $R \subseteq \reals^{i-1}$ where $p$ is delineable, then $\iroot{x_{i}}{p}{j}$ coincides with the root function $\theta_j$ from the above definition on $R$. We thus can evaluate the function computationally by real root isolation.
Beginning from \Cref{alg:line:loop}, we compute projection polynomials whose order-invariance on the underlying cell $R=\Isymb_1\times\cdots\times\Isymb_{i-1}$ (computed in the following iterations) maintain the desired properties of the polynomials in $P_i$.  
In \Cref{alg:line:incomplete} the algorithm might fail, as McCallum's projection operator cannot reason about cells where a polynomial is nullified \cite{mccallum1998}. In \Cref{alg:line:nonnull}, we prevent polynomials from nullifying on any point in the constructed underlying cell by ensuring that at least one coefficient remains non-zero to meet the requirements of the stated theorems (see \cite{nalbach2024levelwise}). \Cref{alg:line:del} maintains delineability of each $p \in P_i$ on $R$, and order-invariance in each of its sections over $R$.
The ordering determined in \Cref{alg:line:irordering} defines a set of resultants to maintain sign-invariance of each $p\in P_i$ in $R \times \Isymb_i$; for this, we analyse $\theta_1,\ldots,\theta_k$ to choose a ``good'' set, possibly exploiting transitivity.
Further, we recall that the cell $\Isymb_1\times\cdots\times\Isymb_n$ is an analytic submanifold of $\reals^n$ as it is bounded by root functions which are analytic by \Cref{thm:delineable}.

\begin{figure}[tb]
    \centering
    \includegraphics[width=15em]{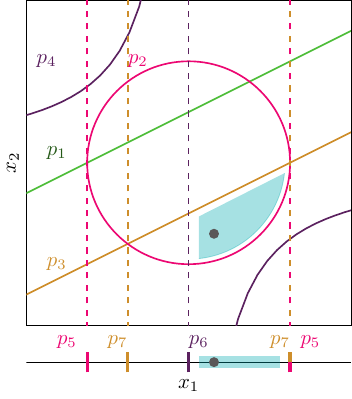}
    \caption{The single cell construction for \Cref{ex:scc}.}
    \label{fig:scc}
\end{figure}

\begin{figure}[tb]
    \centering
    \begin{tikzpicture}
        \draw[help lines, color=gray!30, dashed] (-2.4,-1.9) grid (1.9,1.9);
        \draw[->,thick] (-2.5,0)--(2,0) node[right]{$x$};
        \draw[->,thick] (0,-2)--(0,2) node[above]{$y$};

        \draw[blue100] (0,-1) circle (1);
        \draw[blue100] (-2.5,1) -- (2,1);
        \draw[blue100,fill] (0,0) circle (0.05);

        \draw[purple100,fill] (-2,1) circle (0.05);
        \draw[purple100,fill] (-1,1) circle (0.05);
        \draw[purple100,fill] (1.5,1) circle (0.05);
        \draw[purple100,dotted,thick] (-2.5,1.25) node[above] {\scriptsize $t_3$} -- (2,-1) ;
        \draw[purple100,dotted,thick] (-2,2) node[above] {\scriptsize $t_1$} -- (2,-2) ;
        \draw[purple100,dotted,thick] (2,1.33) node[below] {\scriptsize $t_2$} -- (-2.5,-1.66) ;
        \draw[purple100,fill] (0.8,-0.4) circle (0.05);
        \draw[purple100,fill] (1,-1) circle (0.05);
        \draw[purple100,fill] (-.91,-.6) circle (0.05);

        \draw[darkgreen100,fill] (1.5,1.5) circle (0.05) node[left]{\scriptsize $(a,b)$};
        \draw[darkgreen100,fill] (1,1) circle (0.05);
        \draw[darkgreen100,fill] (1,0) circle (0.05) node[below] {\scriptsize $m$};
        \node[darkgreen100] at (0.75,1.2) {\scriptsize $(m,1)$};
        \draw[darkgreen100,dashed,thick] (2,2)node[left]{\scriptsize $x=my$} -- (-2,-2);
        \draw[darkgreen100,fill] (-1,-1) circle (0.05) node[left] {\scriptsize $(-\frac{2m}{1+m^2},-\frac{2}{1+m^2})$};

        \node at (0,-3) {};
    \end{tikzpicture}
    \caption{Embedding of $\reals$ into $\projs$, and an identification with a circle. $(a:b)$ is identified with the dashed green line, the real number $m$ (consider the intersection with $y=1$), and a point on the unit circle at $(0,-1)$. $t_1$, $t_2$, and $t_3$ are elements of $\projs$.}
    \label{fig:preliminaries:projs}
\end{figure}
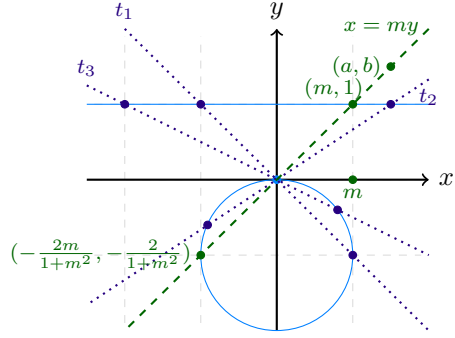

\begin{example}
    \label{ex:scc}
    Consider the polynomials $p_1 = 0.5 x_1 + 0.5-x_2$, $p_2 = x_1^2+x_2^2-1$, $p_3 = 0.5 x_1 - 0.5 -x_2$, $p_4 = -x_1 x_2 -0.75$ as depicted in \Cref{fig:scc}, along with the sample point $s=(0.25,-0.7)$ and a cell as constructed using \Cref{alg:scc}.

    The algorithm adds coefficients and discriminants to guarantee delineability of these polynomials; in this case, we add $p_5 = \disc{x_2}{p_2}$ and $p_6 = \ldcf{x_2}{p_4}$ (the others are trivial). Further, as the cell is described by $\Isymb_2=(\iroot{x_2}{p_2}{1},\iroot{x_2}{p_3}{1})$ on level $2$ we aim to maintain $\iroot{x_2}{p_4}{1} < \iroot{x_2}{p_2}{1}$ on $R$, $\iroot{x_2}{p_2}{1} < \iroot{x_2}{p_3}{1}$ on $R$, $\iroot{x_2}{p_3}{1} < \iroot{x_2}{p_1}{1}$ on $R$, and $\iroot{x_2}{p_3}{1} < \iroot{x_2}{p_2}{2}$ on $R$. We thus add the corresponding resultants, of which only $p_7 = \res{x_2}{p_2}{p_3}$ is non-trivial.
    On level $1$, we determine $\Isymb_1=(\iroot{x_1}{p_6}{1},\iroot{x_1}{p_5}{2})$ as describing the interval.
\end{example}

\section{Projective Delineability}
\label{sec:ProjDel}

We now summarize the theory of projective delineability introduced in \cite{michel2024}.

\paragraph{Real Projective Line.}

Roughly speaking, the real projective line $\projs$ is defined by adding a single point $\infty$ to the real line, so $\projs=\reals\cup\{\infty\}$. We can add more structure to $\projs$ or visualize it by using alternative definitions: identifying the real number $m$ with the line $x=my$ and $\infty$ with $y=0$, we see that the real projective line $\projs$ is the set of lines of $\reals^2$ passing through the origin. Such a line is determined by any of its non-zero vectors $(a,b)\in\reals^2$ or by any of its non-zero multiples, so if we denote by $(a:b)$ the set (equivalence class) of such vectors, we have $\projs=\{(a:b):(a,b)\in\reals^2\setminus\{(0,0)\}\}$. This set identifies with $\reals\cup\{\infty\}$ by mapping $(a:b)$ to $\frac{a}{b}$ if $b\neq 0$ and to $\infty$ otherwise. Finally, $\projs$ identifies with a circle (as an analytic manifold). Possible identifications are visualized in \Cref{fig:preliminaries:projs}.

As $\projs$ identifies with a circle, we cannot use a linear order on $\projs$; however, $\projs$ has a \emph{cyclic ordering} that extends the usual order on $\reals$, as intuitively given in \Cref{fig:preliminaries:projs}.
For distinct $t_1,t_2,t_3 \in \projs$, we use $[t_1,t_2,t_3]$ to denote that ``after $t_1$, one reaches $t_2$ before $t_3$'' in that cyclic ordering on $\projs$. We use $[t_1,\ldots,t_k]$ for $k>3$ to denote $\forall j<j'<j'' \in \prange{k}.\; [t_{j},t_{j'},t_{j''}]$.

\paragraph{Projective Roots.}  

The introduction of the projective line enables us to handle roots at infinity of (univariate) polynomials, and their multiplicities (see \cite[Defn. 2, 3 and 5]{michel2024}): if $p \in \rationals[x]$ has degree less than or equal to $d \in \naturals$, we associate with $p$ the homogeneous bivariate polynomial $H^d(p)$ (also called a binary form of degree $d$) defined by $H^d(p)(x,y)=y^dp(\frac{x}{y})$. The concepts of roots and multiplicities are well-defined for binary forms, and we thus import them for polynomials: $(a:b) \in \projs$ is a {projective root of multiplicity $k$ of $p$ with respect to $d$} if $(a,b)$ is a root of $H^{d}(p)$ with multiplicity $k$.

The set of projective roots of $p$ (with respect to $d$) splits into the real roots on the one hand and the root at infinity on the other hand (see \cite[Lemmas 2 and 3]{michel2024}): $(a:b) \in \projs$ is a projective root of $p$ of multiplicity $k$ w.r.t. $d$ if and only if either $b\neq 0$ and $\frac{a}{b}$ is a real root $p$ of multiplicity $k$ or $b=0$ and $k=d-\deg(p)$.

\paragraph{Projective Delineability.}

We finally formalize the notion of projective delineability, by transferring the concept of projective roots to multivariate polynomials.

\begin{definition}[{Projective Delineability \cite[Defn. 11]{michel2024}}]
    \label{def:pdel}
    Let $i \in \naturals$, $R \subseteq \reals^i$ be a cell, and $p \in \rationals[x_1, \ldots, x_{i+1}] \setminus \{0\}$.
    The polynomial $p$ is called \emph{projectively delineable} on $R$ if there exist continuous functions $\theta_1, \ldots, \theta_k: R \to \projs$ (for $k \in \naturals$) such that:
    \begin{itemize}
        \item for any point in $R$, the values of $\theta_1, \ldots, \theta_k$ are distinct;
        \item the projective roots of the univariate polynomial $p(r, x_{i+1})$ with respect to $\vdeg{x_{i+1}}{p}$ are ${\theta_1(r), \ldots, \theta_k(r)}$ for all $r \in R$; and
        \item there exist constants $m_1, \ldots, m_k \in \posints$ such that for all $r \in R$ and all $j \in \prange{k}$, the multiplicity of the root $\theta_j(r)$ of $p(r, x_{i+1})$ w.r.t. $\vdeg{x_{i+1}}{p}$ is $m_j$.
    \end{itemize}
    The $\theta_j$ are called \emph{projective root functions} of $p$ on $R$. The cells $R \times \theta_j,\ j\in\prange{k}$ are called \emph{projective $p$-sections} over $R$.

    \emph{Analytic projective delineability} is similar, but $R$ is a connected analytic submanifold of $\reals^i$ and the projective root functions are analytic.
\end{definition}

In particular, the first condition means the function values as points around the unit circle maintain a cyclic ordering.
\Cref{fig:pdel:pdel} illustrates an example of a polynomial that is projectively delineable.

\begin{figure}[tb]
    \begin{subfigure}{0.49\textwidth}
        \centering
        \includegraphics[width=15em]{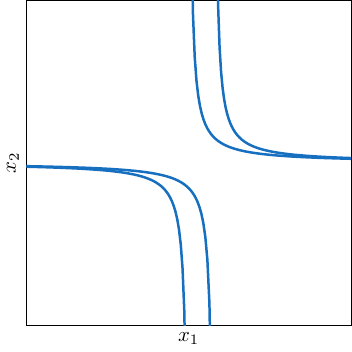}
        
        \caption{Variety of $p$ in $\reals^2$.}
    \end{subfigure}\hfill
    \begin{subfigure}{0.49\textwidth}
        \center
        \includegraphics[width=\textwidth,trim={2cm 0 5cm 0},clip]{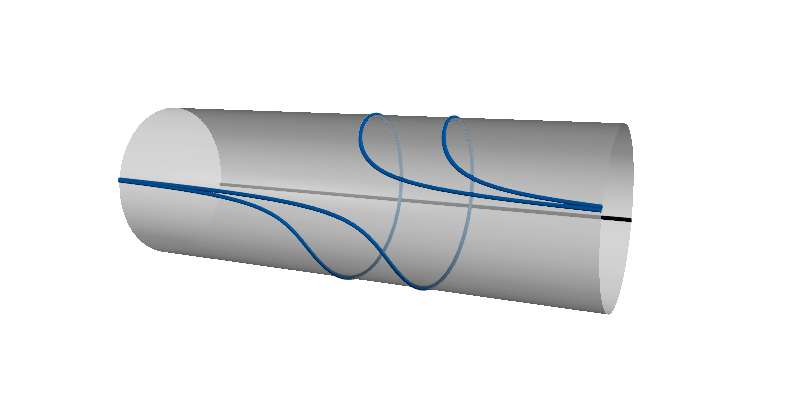}
        \caption{Variety of $p$ in $\reals \times \projs$. The black line marks the points at $\infty$.}
    \end{subfigure}
    \caption{$p=(x_1x_2-1)((x_1-1)x_2-1)$ is projectively delineable on $\reals$, described by two root functions which cross $\infty$.}
    \label{fig:pdel:pdel}
\end{figure}

The central theorem for this work states that order-invariance of the discriminant plus non-nullification is enough to guarantee projective delineability, with no need to maintain the sign-invariance of the leading coefficient.
Note that the theorem requires the underlying cell to be simply connected, which is stronger than connectedness. This assumption is always met for locally cylindrical cells, since they are homeomorphic to open cubes of Euclidean spaces \cite[Proposition 5.3]{basu2006} (which are simply connected).

\begin{theorem}[Projective Delineability {\cite[Thm. 2]{michel2024}}]
    \label{thm:pdel}
    Let $i \in \naturals$, $R \subseteq \reals^i$ be a simply connected analytic submanifold, and $p \in Q[x_1, \ldots, x_{i{+}1}]$ of level $i{+}1$. If $p$ is not nullified on any point in $R$ and $\disc{x_{i{+}1}}{p}$ is not the zero polynomial and is order-invariant on $R$, then $p$ is analytically projectively delineable on $R$ and $p$ is order-invariant in each projective $p$-section over $R$.
\end{theorem}

Delineability is guaranteed by projective delineability plus sign-invariance of the leading coefficient.
\begin{lemma}[Delineability and Projective Delineability {\cite[Cor. 1]{michel2024}}]
    \label{thm:del:pdel}
    Let $i \in \posints$, $R \subseteq \reals^i$ be connected, and $p \in \rationals[x_1, \ldots, x_{i+1}]$ of level $i+1$. Assume that $p$ is projectively delineable on $R$, and $\ldcf{x_{i+1}}{p}$ is sign-invariant on $R$. Then $p$ is delineable on $R$.
\end{lemma}

The definition of projective delineability is equivalent to delineability over cells where the polynomial does not have any roots. In those cases, we can guarantee delineability without sign-invariance of the leading coefficient.

\begin{lemma}
    \label{thm:del:pdelopt}
    Let $i \in \posints$, $R \subseteq \reals^i$, $s \in R$, and $p \in \rationals[x_1, \ldots, x_{i+1}]$ of level $i+1$. If $p$ is projectively delineable on $R$, $\ldcf{x_{i+1}}{p}(s) \neq 0$, and $\realRoots{p(s,x_{i+1})} = \emptyset$, then $p$ is delineable on $R$.
\end{lemma}

\begin{proof}
    Let $\theta_1,\ldots,\theta_k$ be the projective $p$-sections over $R$. If $\ldcf{x_{i+1}}{p}(s) \neq 0$, then it holds $\vdeg{x_{i+1}}{p(s,x_{i+1})} = \vdeg{x_{i+1}}{p}$, thus the projective roots $p(s,x_{i+1})$ are all real. Thus $\{\theta_1(s),\ldots,\theta_k(s)\} = \realRoots{p(s,x_{i+1})} = \emptyset$, and it follows $k=0$. It follows that $p(r,x_{i+1}) \neq 0$ for all $r \in R$, and thus $p$ is delineable on $R$. 
\end{proof}

\paragraph{Root Orderings.}  
As single cell construction relies on root orderings, we give the analogous statement for projective delineability. Note that in contrast to \Cref{thm:irordering}, we can only ensure that two root functions are disjoint or equal.

\begin{theorem}[Lifting of Pairs of Polynomials {\cite[Theorem 3]{michel2024}}]
    \label{thm:ordering:pdel}
    Let $i \in \naturals$, $R \subseteq \reals^{i}$ be a connected analytic submanifold, $s \in R$, and $p_1,p_2 \in \rationals[x_1,\ldots,x_{i+1}]$ of level ${i+1}$. Assume $p_1$ and $p_2$ are analytically projectively delineable on $R$ and $\res{x_{i+1}}{p_1}{p_2}$ is not the zero polynomial and is order-invariant on $R$.  
    Let $\theta_1, \theta_2: R \to \projs$ be real projective root functions of $p_1$ and $p_2$ respectively, and $\sim \in \{ =,\neq \}$ such that $\theta_1(s) \sim \theta_2(s)$. Then $\theta_1 \sim \theta_2$ on $R$.
\end{theorem}

\section{Projective Delineability in Single Cell Construction}
\label{sec:main}

We motivate the use of projective delineability by considering \Cref{ex:scc} again.

\begin{figure}[t]
    \begin{subfigure}[h]{0.49\textwidth}
        \centering
        \includegraphics[width=15em]{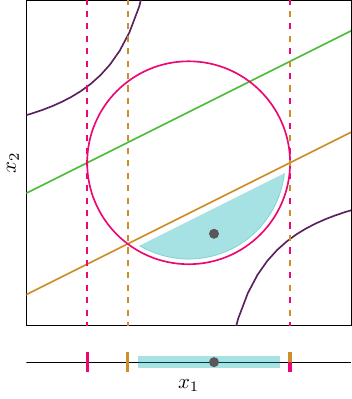}
        \caption{The cell which we aim to construct.}
        \label{fig:motivation:cell}
    \end{subfigure}\hfill
    \begin{subfigure}[h]{0.49\textwidth}
        \centering
        \begin{tikzpicture}
            \node at (0,0) {\includegraphics[width=17em,trim={2.25cm 0 5.75cm 0},clip]{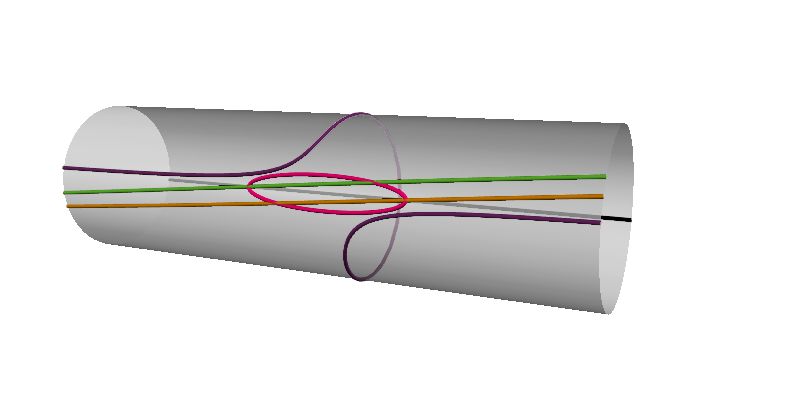}}; %
            \node at (0,-1) {$x_1$};
            \node at (-3.2,0.5) {$x_2$};
        \end{tikzpicture}
        \caption{The varieties in $\reals \times \projs$. The $x_2$ axis is thus cyclically ordered.}
        \label{fig:motivation:proj}
    \end{subfigure}
    \caption{Motivational example for projective delineability.}
    \label{fig:motivation}
\end{figure}

\begin{example}
    The singularity of $p_4$ (witnessed by its leading coefficient) was a boundary to the cell in \Cref{fig:scc}, but crossing that boundary does not change the sign of any input polynomial.  \Cref{fig:motivation:cell} shows the cell that we aim to construct instead: if we detect that the singularity of $p_4$ does not affect the cell, we can omit the leading coefficient of $p_4$ and build the enlarged version.

    For this reasoning, we view the roots of the polynomials in the projective real line, as depicted in \Cref{fig:motivation:proj}, where intuitively $-\infty$ and $+\infty$ are identified with the same point $\infty$. Above the singularity of $p_4$, the two distinct root functions of $p_4$ coincide at the point $\infty$, and thus can be described as a unique real projective root function of $p_4$. The order-invariance of the discriminant of $p_4$ ensures that the variety of $p_4$ can be described using such projective root functions (\Cref{thm:pdel}).    
    Now, by adding the resultant of $p_2$ and $p_4$, we ensure that the mentioned  root function does not intersect the circle, and thus does not enter the cell (\Cref{thm:ordering:pdel}). By projective delineability of $p_4$, we know that there are no other roots (\Cref{def:pdel}), and we thus can omit the leading coefficient from our projection polynomials.
\end{example}

\begin{figure}[tb]
    \begin{subfigure}[t]{0.49\textwidth}
        \centering
        \includegraphics[width=17em,trim={2.25cm 3cm 5.75cm 3.75cm},clip]{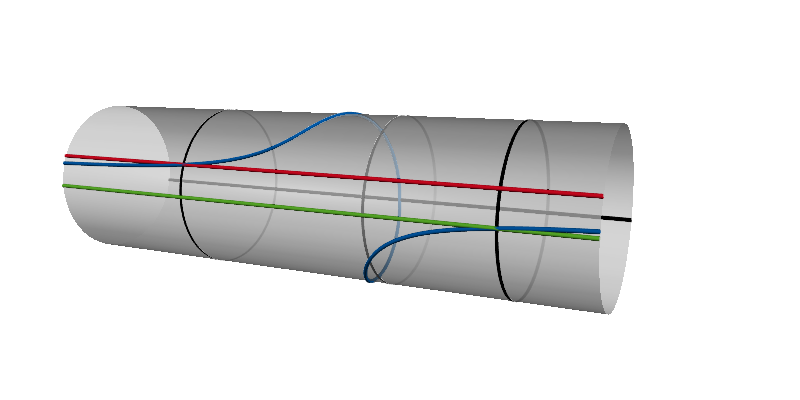}
        \caption{Resultant with bounds.}
        \label{fig:example:simple-resultant}
    \end{subfigure}
    \begin{subfigure}[t]{0.49\textwidth}
        \centering
        \includegraphics[width=17em,trim={2.25cm 3cm 5.75cm 3.75cm},clip]{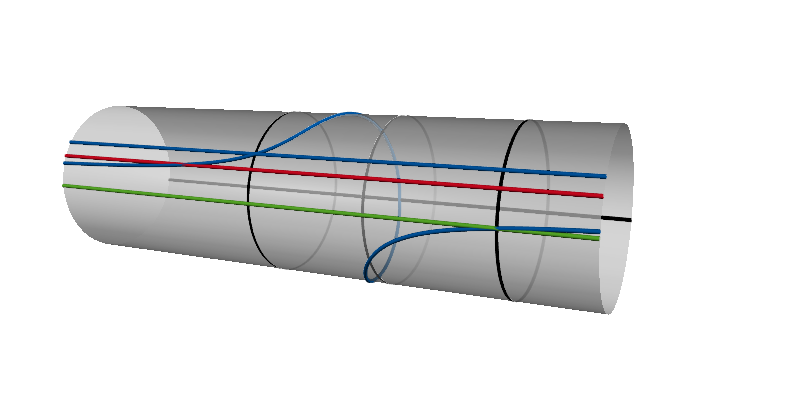}
        \caption{Resultants computed anyway.}
        \label{fig:example:simple-tworoots}
    \end{subfigure}
    \begin{subfigure}[t]{0.49\textwidth}
        \centering
        \includegraphics[width=17em,trim={2.25cm 3cm 5.75cm 3.75cm},clip]{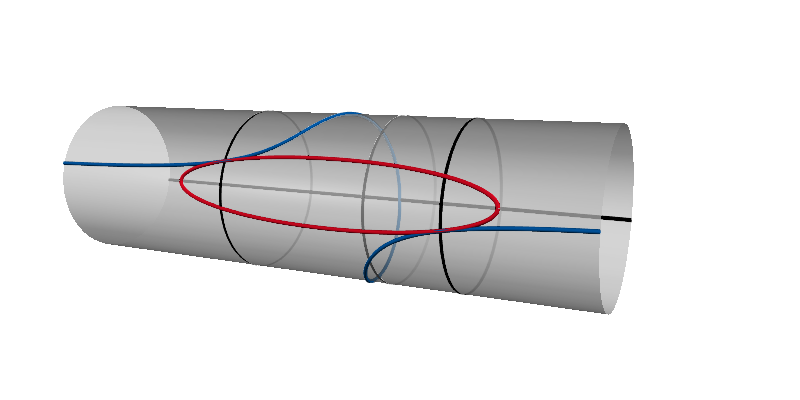}
        \caption{Cell defined by single polynomial.}
        \label{fig:example:simple-samebound}
    \end{subfigure}
    \begin{subfigure}[t]{0.49\textwidth}
        \centering
        \includegraphics[width=17em,trim={2.25cm 3cm 5.75cm 3.75cm},clip]{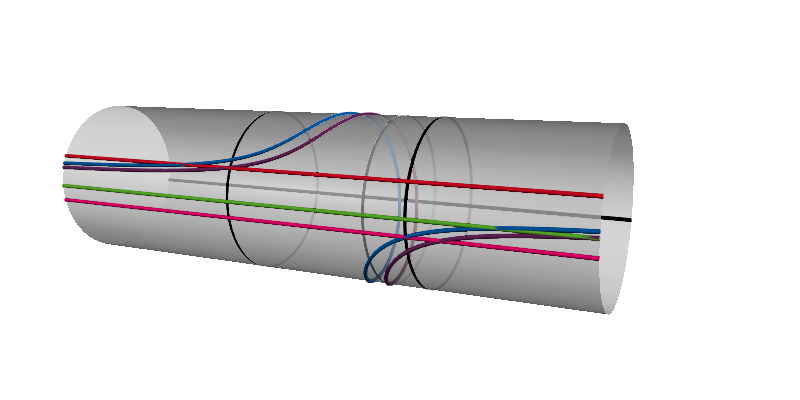}
        \caption{Exploiting transitivity.}
        \label{fig:example:transitive}
    \end{subfigure}
    \begin{subfigure}[t]{0.49\textwidth}
        \centering
        \includegraphics[width=17em,trim={2.25cm 3cm 5.75cm 3.75cm},clip]{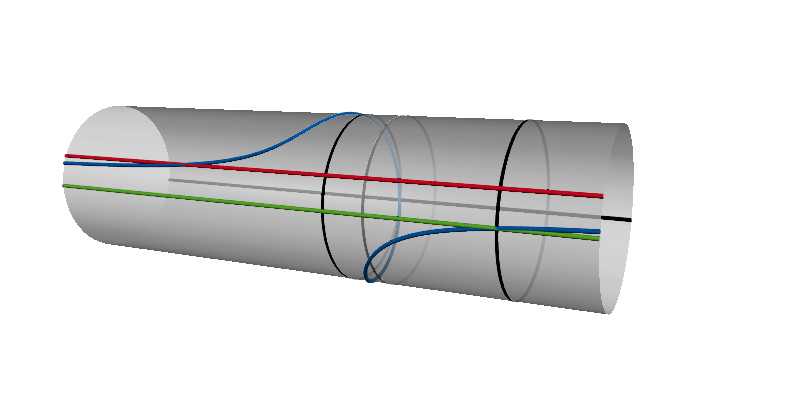}
        \caption{Adding the leading coefficient.}
        \label{fig:example:ldcf}
    \end{subfigure}
    \begin{subfigure}[t]{0.49\textwidth}
        \centering
        \includegraphics[width=17em,trim={2.25cm 3cm 5.75cm 3.75cm},clip]{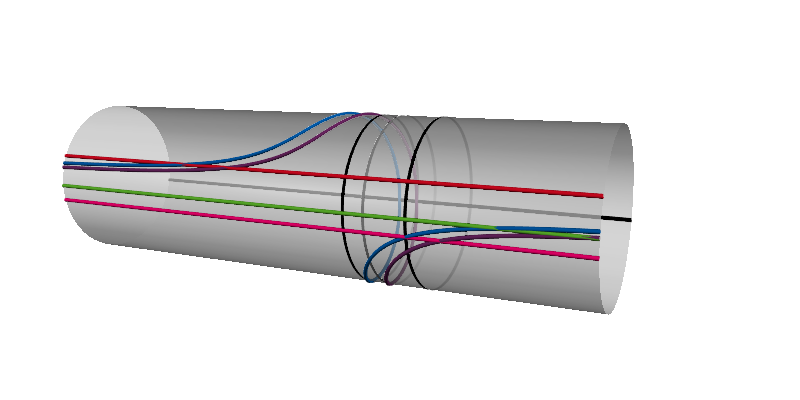}
        \caption{Combining the previous two ideas.}
        \label{fig:example:transitive-ldcf}
    \end{subfigure}
    \caption{Projective root ``orderings'' that guarantee sign-invariance. The grey line indicates the first coordinate  of the sample point. We aim to describe the cell bounded by the green and red polynomials. The black lines indicate the resulting cell boundaries on the first dimension.}
\end{figure}

We modify the single cell construction algorithm as follows. We still describe the cell boundaries using real root functions, as we can encode them using indexed roots which we can evaluate in a straightforward way. Thus, their defining polynomials are still required to be delineable. For the other polynomials, however, we only maintain projective delineability, and hence allow their roots to go through the point at $\infty$.  We thus need to adapt \Cref{thm:invariant}, as a root function may not stay below (above) the lower (upper) bound even if it does not cross it by going through $\infty$.

\begin{theorem}[Projective Root ``Ordering'' for Sign Invariance]
    \label{thm:invariant:pdel}
    Let $i \in \naturals$, $R \subseteq \reals^{i}$ be connected, $s \in R$, and $p,p_\ell,p_u \in \rationals[x_1, \ldots, x_{i+1}]$ of level $i+1$.
    Assume that $p$ is projectively delineable on $R$, $p_\ell$, $p_u$ are delineable on $R$. Let $\theta_\ell,\theta_u: R \to \reals$ be real root functions of $p_\ell$ and $p_u$ on $R$ respectively.

    \begin{itemize}
        \item If $\theta_\ell < \theta_u$ on $R$, and for each projective root function $\theta$ of $p$ on $R$ either
        $\theta = \theta_\ell$ on $R$,
        $\theta = \theta_u$ on $R$, or 
        $\theta_\ell \neq \theta \neq \theta_u$ on $R$ and $[\theta_\ell(s),\theta_u(s),\theta(s)]$;
        then $p$ is sign-invariant on $R \times (\theta_\ell,\theta_u)$.
        \item If for each projective root function $\theta$ of $p$ on $R$ either
        $\theta = \theta_u$ on $R$, or
        $\infty \neq \theta \neq \theta_u$ on $R$ and $[\infty,\theta_u(s),\theta(s)]$;
        then $p$ is sign-invariant on $R \times (-\infty,\theta_u)$.
        \item If for each projective root function $\theta$ of $p$ on $R$ either
        $\theta = \theta_\ell$ on $R$, or
        $\theta_\ell \neq \theta \neq \infty$ on $R$ and $[\theta(s),\theta_\ell(s),\infty]$;
        then $p$ is sign-invariant on $R \times (\theta_\ell,+\infty)$. 
        \item If for each projective root function $\theta$ of $p$ on $R$ either
        $\theta \neq \theta_\ell$ on $R$, or
        $\theta = \theta_\ell$ on $R$;
        then $p$ is sign-invariant on $R \times \theta_\ell$.
        \item If there is no projective root function $\theta$ of $p$ on $R$; then $p$ is sign-invariant on $R \times \reals$.
    \end{itemize}
\end{theorem}

\begin{example}
    This theorem only allows to omit leading coefficients if the interval is bounded in both directions (otherwise we need the leading coefficient for preventing crossing $\infty$) and the resultant of the polynomial with the polynomials defining the lower and upper bound are computed: e.g. we can omit the leading coefficient of blue polynomial in \Cref{fig:example:simple-resultant} (which we would add for its delineability in the classical setting) by additionally adding the resultant of the blue and red polynomial (which is not added in the classical setting). The trade-off may not be attractive; but if the polynomial has a root below and above the cell at the current sample point (e.g. \Cref{fig:example:simple-tworoots}), or if both bounds are defined by the same polynomial (e.g. \Cref{fig:example:simple-samebound}), this would not require additional resultants.
\end{example}

The single cell construction allows for flexibility for different sets of resultants to maintain sign-invariance by exploiting the transitivity of root orderings. In the projective real line, there is only a cyclic ordering. The following insight transfers the idea to the new setting.

\begin{lemma}[Transitive Projective Root ``Ordering'']
    \label{thm:ordering:pdeltrans}
    Let $i \in \naturals$, $R \subseteq \reals^{i}$ be connected, $s \in R$, $\theta_1,\theta_2,\theta_3,\theta_4: R \to \projs$ be continuous. 
    Assume that
    $[\theta_1(s)$, $\theta_2(s)$, $\theta_3(s)$, $\theta_4(s)]$
    and $\theta_1 \neq \theta_2 \neq \theta_3 \neq \theta_{4} \neq \theta_{1}$ on $R$. Then $\theta_1 \neq \theta_3$ on $R$.
\end{lemma}

\begin{example}
    \label{ex:projorderings}
    Consider a projective root (as depicted in blue in \Cref{fig:example:transitive,fig:example:ldcf,fig:example:transitive-ldcf}) that is below the symbolic interval at the current sample point. Preventing the root from crossing the lower bound is analogous to the case of regular delineability. However, we now also need to prevent it to cross the upper bound; for that, we now have three options.  
    (1) We compute a chain of resultants that maintain a cyclic ordering of the given root, the upper bound and some roots in between (in the cyclic sense), exploiting \Cref{thm:ordering:pdel} (in \Cref{fig:example:transitive}, by adding the resultants of the blue with the purple and magenta polynomials respectively, its root is trapped between the other; additionally, we add the resultant of the purple with the red one, and of the magenta with the green).
    (2) We add the leading coefficient of the defining polynomial to the projection, avoiding the intersection of the root with $\infty$, thus making it delineable using \Cref{thm:del:pdel} (see \Cref{fig:example:ldcf}). (3) We mix the two approaches, e.g. for some polynomial in the chain maintaining the cyclic ordering, we add its leading coefficient to avoid crossing $\infty$ (in \Cref{fig:example:transitive-ldcf}, like (1) the blue root is trapped between the magenta and purple, however, instead of the resultant of the purple and red polynomials, we add the leading coefficient of the purple).
\end{example}

\begin{algorithm}[b]
    \caption{Modifications of \texttt{single\_cell\_construction}}
    \label{alg:modifications}
    \DontPrintSemicolon
    \setcounter{AlgoLine}{8}
    \ForEach{$p \in P_i$}{
        \lIf{$p(s_{\prange{i{-}1}},x_i) = 0$}{\Return{FAIL}} 
        $P := P \cup \{ c \}$ for some $c \in \coeff{x_i}{p}$ such that $c(s) \neq 0$ \;
        $P := P \cup \{ \disc{x_i}{p} \}$ \tcp*{proj. delineability, \Cref{thm:pdel}} \label{alg:modifications:pdel}
    }
    $P := P \cup \{ \ldcf{x_i}{p_\ell},\ldcf{x_i}{p_u} \}$ where $p_\ell$,$p_u$ define $\Isymb_i$ \tcp*{del., \Cref{thm:del:pdel}}\label{alg:modifications:del}
    \ForEach(\tcp*[f]{\Cref{thm:del:pdelopt}}){$p \in P_i$ s.t. $\realRoots{p(s_{\prange{i{-}1}},x_i)}=\emptyset$ and $\ldcf{x_i}{p} \neq 0$}{
        $P := P \cup \{ \ldcf{x_i}{p} \}$ \tcp*{delineability, \Cref{thm:del:pdel}}\label{alg:modifications:noroots}
    }
    \textbf{choose} symmetric ${\not\approx} \subseteq (\Theta \cup \{\infty\})^2$ that ensures $\theta_\ell \neq \theta_u$ (if $\Isymb_i = (\theta_\ell,\theta_u)$) and  sign-invariance of each $p \in P_i$ by \Cref{thm:invariant:pdel}, using ``transitivity'' by \Cref{thm:ordering:pdeltrans} \; \label{alg:modifications:ordering}
    $P := P \cup \{ \res{x_i}{p}{p'} \mid (\iroot{x_i}{p}{j},\iroot{x_i}{p'}{j'}) \in {\not\approx}\}$ \tcp*{\Cref{thm:ordering:pdel}} \label{alg:modifications:res}
    $P := P \cup \{ \ldcf{x_i}{p} \mid (\iroot{x_i}{p}{j},\infty) \in {\not\approx} \}$ \tcp*{\Cref{thm:del:pdel}} \label{alg:modifications:ldcf}

\end{algorithm}

To summarize, we replace Lines~\ref{alg:line:mod-begin}~to~\ref{alg:line:mod-end} of \Cref{alg:scc} by \Cref{alg:modifications}. We maintain projective delineability for each polynomial (\Cref{alg:modifications:pdel}) using \Cref{thm:pdel}, and delineability using \Cref{thm:del:pdel} for polynomials defining the bounds (\Cref{alg:modifications:del}) and for polynomials without roots (\Cref{alg:modifications:noroots}) where we make use of the optimization by \Cref{thm:del:pdelopt} to omit leading coefficients. 
In \Cref{alg:modifications:ordering}, we determine a relation for maintaining sign-invariance of polynomials (\Cref{thm:invariant:pdel}), where we may exploit transitivity (using \Cref{thm:ordering:pdeltrans}), now maintaining a cyclic ordering. The relation involves $\infty$ to enable the choice of adding leading coefficients (\Cref{alg:modifications:ldcf}) using \Cref{thm:del:pdel} instead of resultants (\Cref{alg:modifications:res}) using \Cref{thm:ordering:pdel}.

We now elaborate how me make determine the ordering in \Cref{alg:modifications:ordering}, exploiting the options elaborated in \Cref{ex:projorderings}.
\cite{nalbach2024levelwise} describes some heuristics for choosing root orderings (in the non-cyclic setting). The \textsc{biggest cell} heuristic is the straightforward ordering that fulfils the minimal requirement from \Cref{thm:invariant}, and the \textsc{lowest degree barrier} heuristic greedily minimizes the degrees of the computed resultants using transitivity. 
In our modification, our aim is - compared to the classical setting - to avoid additional resultants but omit leading coefficients whenever possible. We thus compute the ordering according to the \textsc{biggest cell} or \textsc{lowest degree barrier} heuristic as in the classical setting (as if every polynomial would be delineable). To transfer this ordering to the projective delineability setting, for each root $\theta$ above (below) the cell, we either add $(\theta,\infty)$ (effectively adding a leading coefficient) or $(\theta,\theta')$ for some root $\theta'$ below (above) the cell (effectively adding a resultant with a polynomial that has roots on the ``other'' side). As we do not want to add additional resultants, we do the latter only if the corresponding resultant would have been added in the non-projective case.

\section{Experiments}
\label{sec:experiments}

We implemented our single cell construction algorithm based on the proof system described in \cite{nalbach2024levelwise} in our solver \texttt{SMT-RAT}, which uses it for generating explanations for the NLSAT algorithm. For this paper, we test the following variants: \texttt{BC} and \texttt{LDB} are the baseline variants using the \textsc{biggest cell} and \textsc{lowest degree barriers} heuristic respectively. \texttt{BC-PD} and \texttt{LDB-PD} are the modified versions using projective delineability as described above. Although we use the incomplete McCallum's projection operator, the implementation of our proof system is complete: in case a polynomial is nullified, we add some of its partial derivatives to ensure its order invariance, as suggested in \cite[Section 5.2]{mccallum1985}.

We conduct our experiments on Intel\textregistered Xeon\textregistered 8468 Sapphire CPUs with 2.1 GHz per core, testing upon the SMT-LIB \emph{QF\_NRA} benchmark set \cite{preiner2024smtlib} which contains $12\,154$ instances.  We use a time limit of $60$ seconds and a memory limit of $4$ gigabytes.  The code, instructions for reproducing and raw results are available at: \href{https://doi.org/10.5281/zenodo.14900915}{\url{https://doi.org/10.5281/zenodo.14900915}}.

\begin{table}[b]
    \caption{Number of instances solved by \texttt{BC} resp. \texttt{LDB}.}
    \label{tbl:results}

    \centering
    \begin{tabular}{lrr}
        \toprule
         & \texttt{BC} & \texttt{LDB} \\
        \midrule
        solved by no variant & 1977 & 1978 \\
        solved by \texttt{BC}/\texttt{LDB} but not by \texttt{BC-PD}/\texttt{LDB-PB} & 20 & 16 \\
        solved by \texttt{BC-PD}/\texttt{LDB-PB} but not by \texttt{BC}/\texttt{LDB} & 16 & 18 \\
        solved by both variants & 10\,141 & 10\,142 \\
        \bottomrule
    \end{tabular}
\end{table}

\paragraph{Overall Results.} 
The number of solved instances is reported in \Cref{tbl:results}, showing that the use of projective delineability does not greatly affect which problems are tractable within the time limit.  The actual running times are depicted in \Cref{fig:evaluation:runtime}: they show similar performance of the modified and baseline versions on the majority of instances, but significantly different behaviour on some instances.  The differences largely even out over the whole benchmark set (a typical picture for changes to the projection heuristics in our experience). Nevertheless, it is clear that many instances do benefit from projective delineability. Identifying a criterion to predict whether the optimization pays off \emph{a priori} is desirable (a machine learning based approach may be possible \cite{delrio2024}), but that is not in the scope of this paper.

In the remainder of this section we compare the behaviour of \texttt{BC} and \texttt{BC-PD} variants to better understand these results.

\paragraph{Number of Applications.}
The results may suggest that the optimization from projective delineability is only applied in very rare cases, but this is not the case.  Considering the instances solved by \texttt{BC-PD}: in total, the leading coefficient can be omitted for $307\,822$ polynomials, while for $826\,795$ polynomials, the optimization cannot be applied as the cell is unbounded in some direction, and for $4\,089$ polynomials, the optimization cannot be applied as it does not have a root on both sides of the bounds, or we did not find an appropriate resultant that may replace the leading coefficient. The optimization is thus applied in a substantial $37\%$ of the cases. However, we also need to add a coefficient for each polynomial to ensure its non-nullification: if we add the leading coefficient, this suffices in most of the cases. Still, if it is omitted, we may choose another coefficient which is potentially simpler (e.g. of lower degree, fewer variables, etc) that the leading coefficient. We thus investigate whether this choice has an impact on another metric.

\begin{figure}[tb]
    \begin{subfigure}{0.49\textwidth}
        \centering
        \includegraphics[width=15em]{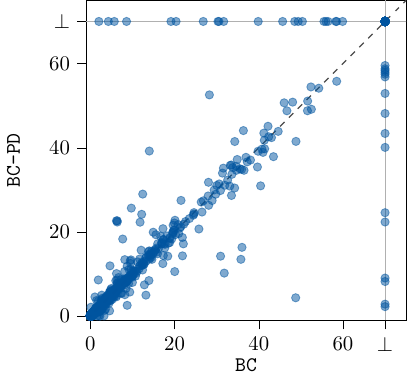}
        \caption{Running time in seconds.}
        \label{fig:evaluation:runtime}
    \end{subfigure}\hfill
    \begin{subfigure}{0.49\textwidth}
        \centering
        \includegraphics[width=15em]{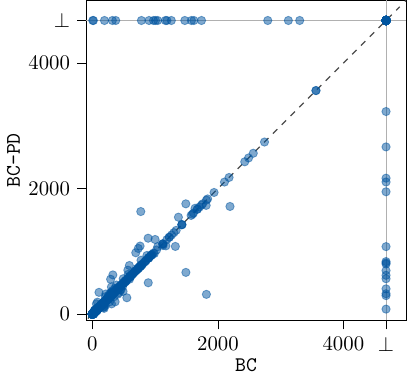}
        \caption{Number of constructed cells.}
        \label{fig:evaluation:num_cells}
    \end{subfigure}\\
    \begin{subfigure}{0.49\textwidth}
        \centering
        \includegraphics[width=15em]{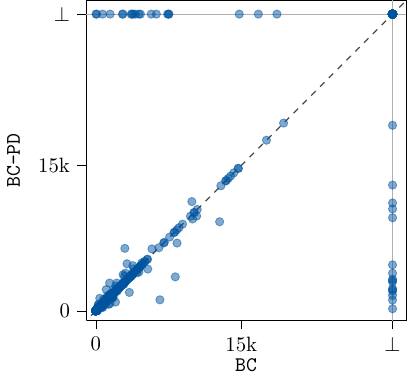}
        \caption{Number of polynomials.}
        \label{fig:evaluation:num_polys}
    \end{subfigure}\hfill
    \begin{subfigure}{0.49\textwidth}
        \centering
        \includegraphics[width=15em]{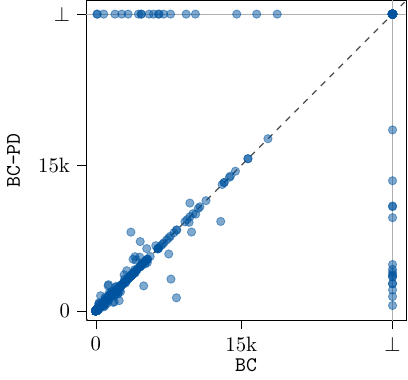}
        \caption{Number of computed roots.}
        \label{fig:evaluation:num_roots}
    \end{subfigure}
    \caption{Comparison of different metrics. Each point indicates an instance. $\bot$ indicates a timeout of the corresponding solver.}
    \label{fig:evaluation}
\end{figure}

\paragraph{Quality of Cells.}
A good indicator for the quality of the generated cells is their size, which may be indirectly measured by the number of cells constructed  (\Cref{fig:evaluation:num_cells}). By leaving out leading coefficients, we also hope to decrease the number of projection polynomials (\Cref{fig:evaluation:num_polys}). By choosing ``simpler'' coefficients than the leading coefficients (in terms of a lower degree), we aim to reduce the degree of the computed polynomials (not depicted in \Cref{fig:evaluation} as there is no visible difference between the variants) and indirectly the number of computed roots (\Cref{fig:evaluation:num_roots}). However, all these measures do not seem to differ significantly in aggregate between the baseline and modified variants. An explanation could be that the (degree of the) leading coefficients do not carry much weight, or that the alternative coefficients are not much simpler.

\paragraph{Role of Projection Polynomials.}
To address the latter hypothesis, we compare the impact of the different projection polynomials (resultants, discriminants, and (leading) coefficients). On the $10\,157$ instances solved by \texttt{BC-PD}, of the time spent on algebraic computations, $55\%$ is spent on computing discriminants and $5\%$ is spent on computing resultants, and almost no time is spent on computing coefficients. Regarding the polynomials with a maximum total degree of the ones occurring in an instance, in $15\%$ of the instances, that polynomial is a discriminant, in $30\%$ of the instances, that polynomial is a resultant, and in only $3\%$ of the instances, that polynomial is a coefficient. Both measures show that the coefficients have a minor impact on the complexity of the projection.

\section{Conclusion}
\label{sec:conclusion}

We provided a first application for the notion of projective delineability recently introduced in \cite{michel2024}, that formally describes the role of the leading coefficients in the projection. We modified the single cell construction algorithm accordingly and evaluated the result in the context of the NLSAT algorithm for SMT solving. The results offer a variety of possibilities for modifying the projection, and thus fit nicely as an extension of the proof system introduced in \cite{nalbach2024levelwise}.

Our experimental evaluation shows the resulting optimization is applied in many cases, however, these do not translate to significant improvements in terms of running times or quality of intermediate results, as other projection polynomials play a greater role for the computational effort.  That is the case when viewing the dataset as a whole:  we have found many individual instances that benefit from the optimisation bringing further the future research question of how to recognise this in advance.  Also in future work, further symbiotic optimizations may lead to practical improvements, such as reducing the amount of coefficients required for maintaining non-nullification of polynomials.

\printbibliography

\end{document}